\documentclass[12pt]{article}
\usepackage[a4paper,left=2.5cm,right=2.5cm,top=2.5cm,bottom=3.0cm]{geometry}

\usepackage[T1]{fontenc}
\usepackage[dvipsnames]{xcolor}
\usepackage{graphicx}
\usepackage{epstopdf} 
\usepackage{enumitem}
\usepackage[pdfusetitle=false,bookmarks=false]{hyperref}
\usepackage[normalem]{ulem}
\usepackage{amsthm}
\usepackage{amsmath}
\usepackage{amsfonts}
\usepackage{amssymb}
\usepackage[noend, noline, ruled, linesnumbered]{algorithm2e}

\begin{document}

\title{A Framework for Distributed Discrete Evacuation Strategies}

%\author[1]{Piotr Borowiecki}
%\ead{p.borowiecki@issi.uz.zgora.pl}
%%\orcid{0000-0002-5239-6540}
%
%\affiliation[1]{
%organization={Institute of Control and Computation Engineering, University of Zielona G\'ora},
%city={Zielona G\'ora},
%country={Poland}
%}
%
%\author[2]{Dariusz Dereniowski}
%%\authornotemark[1]
%%\orcid{0000-0003-4000-4818}
%\ead{deren@eti.pg.edu.pl}
%
%\affiliation[2]{
%    organization={Gda\'nsk University of Technology},
%    city={Gda\'nsk},
%    country={Poland}
%}
%
%\author[3]{{\L}ukasz Kuszner}
%%\orcid{0000-0003-1902-7580}
%\ead{lukasz.kuszner@ug.edu.pl}
%
%\affiliation[3]{
%    organization={Institute of Informatics, University of Gda\'nsk},
%    city={Gda\'nsk},
%    country={Poland}
%}

\author{
     Piotr Borowiecki{\footnotemark[1]}
\and Dariusz Dereniowski{\footnotemark[2]} 
\and {\L}ukasz Kuszner{\footnotemark[3]}
}

\date{}

\maketitle

\def\thefootnote{\fnsymbol{footnote}}

\footnotetext[1]{Institute of Control and Computation Engineering, University of Zielona G\'ora, Poland}

\footnotetext[2]{Gda\'{n}sk~University~of~Technology, Poland}

\footnotetext[3]{Institute of Informatics, University of Gda\'nsk, Poland}

% ----------------------------------------------------------

\newtheorem{theorem}{Theorem}[section]
\newtheorem{lemma}[theorem]{Lemma}
\newtheorem{proposition}[theorem]{Proposition}

\newcommand{\loc}[2][]{\if!#1! \pi_{#2}\else \pi_{#2}(#1)\fi}
\newcommand{\B}{B}
\newcommand{\OPT}{OPT}
\newcommand{\dist}[2]{\mathrm{dist}(#1,#2)}
\newcommand{\agents}{\mathcal{A}}
\newcommand{\bigo}{\mathcal{O}}   
\newcommand\YES{\mbox{\normalfont\textsc{yes}}}
\newcommand\NO{\mbox{\normalfont\textsc{no}}}
\newcommand{\problemEvac}{\mbox{\normalfont\textsc{Evac}}}
\newcommand{\ZG}{\mathcal{G}}
\newcommand{\ZGV}{\mathcal{V}}
\newcommand{\ZGE}{\mathcal{E}}
\newcommand{\cZ}{\mathcal{Z}}
\newcommand{\cS}{\mathcal{S}}

\begin{abstract}
In this paper, we study discrete evacuation in networks, where agents know the network topology and designated exit nodes but do not know the number and initial positions of other agents. Each agent initially occupies a distinct node and must reach any exit node.
Operating in a synchronous distributed model with local communication, the agents aim to minimize the time when the last agent reaches an exit.
We introduce a general algorithmic framework for constructing evacuation strategies on arbitrary graphs. 
As a key application, we demonstrate that the framework yields asymptotically optimal evacuation strategies---achieving a constant competitive ratio---for grid networks, with natural extensions to triangular and hexagonal grids.
\end{abstract}


% ================================================================
\section{Introduction}

Evacuation is a fundamental coordination task for mobile agents operating in networks under a distributed model.
While most existing research has focused on continuous, flow-like, or geometric models, the \emph{discrete evacuation} problem---where agents traverse vertices of a graph and must reach exits relying only on local communication and limited initial information---has received considerably less systematic attention from the perspective of distributed algorithms.
This paper addresses this gap by showing that discrete evacuation admits a unifying \emph{distributed framework} with algorithmic guarantees on evacuation time.

%--------------------------------------------------
\subsection{Motivation} \label{sec:motivation}

From a theoretical perspective, the centralized (offline) version of discrete evacuation can be solved efficiently via network flow techniques~\cite{BorDasDerKusz25}. 
This raises a natural question of how much harder does the problem become when agents operate in a distributed (online) setting with limited information. 
We address this by comparing distributed strategies---where agents \emph{neither know the number nor the initial positions of others}---to their centralized counterparts, highlighting the sensitivity of evacuation to input knowledge. 
In our model, agents \emph{know the graph and exit locations} in advance.
Without this a priori information, the task generalizes classical distributed problems such as treasure hunt (see, e.g., Bouchard et al.~\cite{BouchardDLP21}) and graph exploration (see, e.g., Kolendarska et al.~\cite{KolenderskaKMZ09}), for which constant-competitive algorithms are known not to exist. A similar barrier arises when agents know the graph but not the exit locations.

Beyond theoretical considerations, our distributed approach eliminates the unrealistic assumption that a central planner knows the exact distribution of evacuees a priori---a premise that often fails in modern environments such as high-rise buildings, multilevel transport systems, or sports arenas, where population density and structural complexity make outcomes difficult to predict (see, e.g., Sun et al. \cite{HHZR24}). This is particularly relevant in modern decentralized multi-agent architectures, where evacuation strategies must be computed \emph{on the fly} by autonomous entities (see, e.g., Avil{\'e}s et al. \cite{MoralTK14}).
Such trends highlight the growing relevance of distributed evacuation in human and robotic applications.

%------------------------------------------------------------
\subsection{Related Work} \label{sec:related-work}

A recurring theme in the literature on evacuation problems is that of terrain evacuation, commonly modeled either as continuous (geometric) or discrete (see, e.g., a survey by Sun et al.~\cite{HHZR24}). The continuous line of work, initiated (in distributed form) by Chrobak et al.~\cite{CGGM15} for evacuation of collaborative agents from a line, and extended to multiple rays by Brandt et al.~\cite{BrandtFRW20}, has been extensively studied on planar domains such as discs and related geometric settings, see, e.g.,~\cite{BrandtLLSW17,CGKNOV15,CGGKMP14,DisserS19,GeorgiouLLK23,PattanayakR0S18}. By contrast, the distributed \emph{discrete} evacuation problem studied here has received very limited attention: to our knowledge, the only prior work introducing essentially the same distributed discrete model is the paper by Borowiecki et al.~\cite{BorDasDerKusz25}.

Discrete graph models are natural for environments obtained by discretizing Euclidean space into small cells and representing adjacency via lattice or grid graphs. 
Graph‑based abstractions of terrain and agent coordination are widely used, see, e.g.,~\cite{AltshulerYWB11,BhadauriaKIS12,KaraivanovMSV14,Markov_2016}. 
In particular, two-dimensional grids and their subgraphs (so called partial grids) naturally discretize polygons and other planar searching environments (see, Dereniowski et al.~\cite{DereniowskiO19}). 
Additionally, other topologies are used to model two-dimensional space. For example, a triangular lattice or hexagonal grids are commonly used in the area of programmable matter, see, e.g.,~\cite{Chalopin0K24,DaymudeGHKSR20, HinnenthalLS24,LunaFSVY20,NavarraPBT23}. 
In this context, we note that our approach for grids can be analogously applied to partial grids as well, but determining whether partial grids admit constant‑competitive \emph{evacuation} strategies remains an intriguing open question.

While we focus on static graphs, related work has also examined evacuation and motion planning over \emph{temporal graphs}, which capture time‑varying connectivity and constraints, see, e.g., Akrida et al.~\cite{AkridaCGKS19}, and Erlebach et al.~\cite{ErlebachS23}. This perspective is complementary to ours and could be combined with our framework in future work.

Heterogeneous agent coordination has been explored for search tasks (see, Dereniowski et al.~\cite{DereniowskiKO21}) but also for evacuation problem---see, Borowiecki et al.~\cite{BorDasDerKusz25}, where each type of agent can access only a certain subgraph of the original graph (the problem is NP-hard even if there are only two types of agents, and also if the optimal evacuation time is a small constant).
Evacuation and collaborative search are closely related, but have distinct objectives. Evacuation typically minimizes the time until the \emph{last} agent reaches an exit, with the number of agents given as part of the input. In contrast, collaborative search (see, e.g., Alspach ~\cite{Alspach04}) focuses on minimizing the time for the \emph{first} agent to find a target, often treating the number of agents as an optimization parameter; see also Kappmeier~\cite{Kappmeier15} for alternative objectives such as maximizing rescues within a deadline. 

Most relevant to a \emph{distributed 
setting} in this work, Borowiecki et al.~\cite{BorDasDerKusz25} formalized the distributed discrete evacuation model and proved that there does not exist a~$(2-\Theta(1/k))$-competitive algorithm for evacuating $k$ agents on trees, while also establishing a constant‑competitive strategy for trees.

%--------------------------------------------------
\subsection{Our Contributions} \label{sec:contribution}

In this paper, we consider a distributed model in which the agents know the graph structure and exit locations but lack information regarding their quantity and initial placement. The objective is to minimize the time until the last agent reaches an exit using local, synchronous communication.

We introduce an algorithmic framework for constructing evacuation strategies for arbitrary graphs and establish a general formula (Theorem~\ref{thm:EvacuationTime}) for the evacuation time achieved by any strategy within this framework.
In terms of lower bounds, we prove that any evacuation strategy restricted to a single spanning tree is inherently limited, establishing a lower bound of $\Omega(\sqrt[5]{N})$ on the competitive ratio, where $N$ is the order of a graph (Proposition~\ref{prop:spanning-tree}). 

As a primary application, we show that the proposed framework yields $\bigo(1)$-compe\-titive distributed strategies for grid graphs (Theorem~\ref{thm:grids}). Our approach naturally extends to triangular and hexagonal grids.

Notably, several concepts introduced in the framework---such as the $\B$-partitions, zone graphs, and inter-zone coordination mechanisms---provide a versatile algorithmic toolkit amenable for adaptation to other distributed problems.

%--------------------------------------------------
\subsection{Paper Outline}

The paper is organized as follows. 
Section~\ref{sec:model} formalizes the discrete evacuation model and specifies assumptions regarding agents, communication, and movement.
Sections \ref{sec:problem} and \ref{sec:groups} give a formal problem statement and establish basic properties that enable agent grouping and coordination.
A lower bound for strategies following spanning trees is given in Section~\ref{sec:sptrees}.
In Section~\ref{sec:framework} we develop our evacuation framework, including the \mbox{epoch-phase} structure and the associated algorithms.
Section \ref{sec:framework:analysis} analyzes the framework and proves the evacuation time bound given in Theorem~\ref{thm:EvacuationTime}.
Section~\ref{sec:grids} applies the framework to grids, showing that it yields $\bigo(1)$‑competitive distributed strategies. 
Section~\ref{sec:conclusion} summarizes our contributions and discusses open problems and directions for future research.

% ========================================================================
\section{The Distributed Evacuation Model}
\label{sec:model}

We adopt the distributed evacuation model of~\cite{BorDasDerKusz25}. The input is a finite, simple, undirected graph $G=(V,E)$ with the vertex set $V$, the edge set $E$, and two distinguished nonempty vertex subsets $H$ and $X$, called \emph{homebases} and \emph{exits}, respectively.
Initially, each of the $k$ agents resides in one of the homebases in $H$ such that no two agents occupy the same homebase (i.e., $|H|=k$ and $H\cap X=\emptyset$).
Each agent a priori knows the graph (i.e., it has an isomorphic copy of $G$ in memory), knows the set $X$ of exits and its current position. However, no agent knows $k$ or the set $H$ of homebases.

The model is synchronous, i.e., the time taken by the evacuation process is divided into \emph{steps} of unit duration. 
In each step, each agent: first communicates with the other agents, then performs local computations, and finally decides on its action/movement.

%--------------------------------------------------
\subsection{Communication within a Single Step}

In the model, each agent can directly communicate with another agent when they are located at adjacent vertices, or at the vertices having a common neighbor.
We refer to such communication as \emph{direct}.

The number of messages two agents may exchange at the beginning of a step is finite, but we do not assume any specific bound.
In consequence, if there are two agents located at vertices $u$ and $v$ at a distance greater than $2$, and there is a $uv$-path $P$ whose internal vertices are \emph{occupied} in the sense that every vertex is either occupied by an agent or has both neighbors on the path occupied, then $u$ and $v$ may communicate at the beginning of the step through multiple message exchanges done by the agents located on $P$. This enables gossip propagation along occupied paths. 
We refer to such communication as \emph{indirect}.

% ---------------------------------------------------------
\subsection{Computation and Movement within a Single Step}
\label{sec:compandmove}

After the communication phase, each agent performs a local computation based on its \emph{current state} that is, its memory contents and location.
Such local computation results in determining the \emph{action} performed by the agent in the given step, which is one of the following:

\begin{enumerate}[label=(\Alph*)]
 \item \label{action:1} the agent \emph{moves}, i.e., changes its location from the currently occupied vertex $v$ to one of its neighbors $u$, or
 \item \label{action:2} the agent \emph{stays still}, which means that it remains at the currently occupied vertex.
\end{enumerate}

We assume that all moves occur simultaneously, and no two agents may occupy the same vertex once the moves occur.
More precisely, in any step, action \ref{action:1} is feasible if no other agent decides to move to $u$, and if $u$ is occupied at the beginning of the considered step, then the occupying agent needs to move from $u$ to another vertex $u'$ in the same step.
(Note that we allow $u'=v$, which results in a \emph{swap} of the agents).
On the other hand, action \ref{action:2} is feasible if no other agent decides to move to $v$ in the considered step.
An agent that occupies an exit at the end of a step \emph{evacuates}, i.e., it is automatically removed from the graph.
Hence, at most one agent may evacuate through a given exit per step.

% =====================================================================
\section{The Problem Formulation}
\label{sec:problem}

Let $\loc[s]{i}$ denote the vertex occupied by agent $i$ at the end of a step $s$, with $\loc[0]{i}$ denoting its initial position (\emph{the homebase}). 
A tuple $(\loc{1},\ldots,\loc{k})$ of such positioning functions, where $\loc{i}$ has length $s_i$, is an \emph{evacuation strategy} if for each agent $i$ and for each $s<s_i$, $\loc[s]{i}\notin X$, and no two distinct agents $i,j$ occupy the same vertex at the same step, i.e.,
$\loc[s]{i}\neq\loc[s]{j}$ for all steps $s\le\min\{s_i,s_j\}$. 
Moreover, the evacuation strategy is \emph{successful} if $\loc[s_i]{i}\in X$ for each $i\in\{1,\ldots,k\}$.
The \emph{length} of an evacuation strategy $(\loc{1},\ldots,\loc{k})$ is defined as $\max\{s_1,\ldots,s_k\}$.
A successful evacuation strategy is \emph{optimal} if its length is minimized.

The discrete evacuation problem $\problemEvac$ is defined as follows:

\begin{description}
\item[\textsc{Discrete Evacuation ($\problemEvac$)}]
\item[Input:] A graph $G$, an integer $\ell$, a set $X$ of exits, and a set $H$ of homebases occupied by $k$ agents, where $k=|H|$.
\item[Question:] Is there a successful evacuation strategy of length at most $\ell$?
\end{description}

Since in this work we consider $\problemEvac$ in a \emph{distributed setting}, we assume that the vertices of the graph have unique identifiers, and that the agents have unique identifiers.
At the beginning of each step, each agent can access the identifier of the vertex $v$ it occupies, and for each edge $e$ incident to $v$ the agent is given the identifier of the other end of $e$.
Thanks to this assumption, the agent can correctly perform a move in each  step.
Formally, by a \emph{distributed algorithm} we understand an algorithm executed independently by each agent, which, given the same input, determines the agent’s action in every step.

\medskip
To evaluate the efficiency of the considered strategies, we use the standard measure of competitive analysis: the \emph{competitive ratio} that is the worst-case ratio of the evacuation time achieved by a distributed algorithm to the optimum centralized evacuation time~$\OPT$.

% =========================================================================
\section{Limitations of Spanning-Tree-Based Strategies}
\label{sec:sptrees}

A natural question concerns the best competitive strategies for various network topologies. 
Also, one might wonder whether there is some natural way of obtaining a constant-competitive distributed algorithm based on a fixed spanning tree of a given graph (perhaps also supported by the algorithm developed in \cite{BorDasDerKusz25}).
However, as we show in Proposition \ref{prop:spanning-tree}, any strategy in which agents restrict their movements to a chosen spanning tree fails to achieve a constant competitive ratio.
A natural definition of a ``spanning-tree-strategy'' requires that upon its completion, all edges traversed by the agents result in an acyclic subgraph. In the proposition below, we assume this subgraph is computed in the first step by the agents given the graph and exits locations.

\begin{proposition} \label{prop:spanning-tree}
There exist input instances with $N$ vertices for $\problemEvac$ for which any distributed algorithm that moves the agents along an arbitrarily selected spanning tree has competitive ratio $\Omega(\sqrt[5]{N})$.
\end{proposition}

\begin{figure}[tbp]
\centering
  \begin{minipage}{0.47\textwidth}
  \centering
   \includegraphics[width=\textwidth,clip]{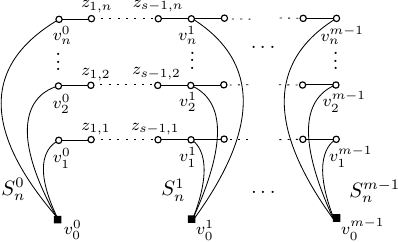}
    %\caption{(a)}
  \end{minipage}
  \hfill % figure distance
  \begin{minipage}{0.47\textwidth}
    \centering
    \includegraphics[width=\textwidth,clip]{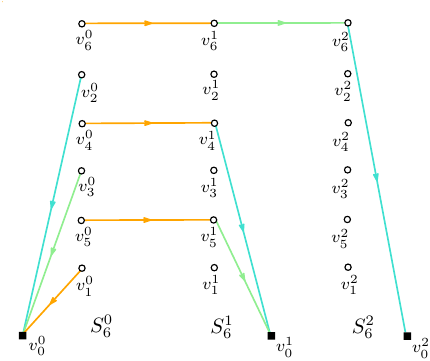}
   %\caption{(b)}
  \end{minipage}
  \caption{(a) The structure of graphs constructed in the proof of Proposition~\ref{prop:spanning-tree}; the exits are marked as black squares (b) An optimal strategy for $k=6$ agents, $m=3$ stars and $s=1$; unused edges are omitted; colors (orange, green, blue) correspond to subsequent steps of the strategy.}
  \label{fig:LBonSpTrees}
\end{figure}

\begin{proof}
Consider a union of $m$ star graphs $S_n^0, S_n^1, \ldots, S_n^{m-1}$ with centers $v_{0}^0, v_{0}^1, \ldots, v_0^{m-1}$, and let $v_{i}^j$ with $i \in \{1, \ldots, n\}$, $j \in \{0, \ldots, m-1\}$ be their leaves.
For each $j \in \{1, \ldots, m-1\}$ and $i \in \{1, \ldots, n\}$, we connect the leaf $v_i^{j-1}$ with the corresponding leaf $v_i^{j}$ by a path of length $s$. In the resulting graph $G$, let the set of exits consist of the centers of the stars, i.e., $X = \{ v_{0}^0, v_0^1, \ldots, v_0^{m-1} \}$.
(See Figure \ref{fig:LBonSpTrees}(a) for an illustration of this structure.)

We set the parameters such that $n=sm(m+1)/2$, $m=\Theta(\sqrt[5]{N})$  and $s=\Theta(m)$, where $N$ is the order of the graph.

In what follows we examine all spanning trees of $G$ by distinguishing two complementary cases.  
In each case we give an example of an initial placement of agents that results in sub-optimal evacuation time claimed in the proposition.

\medskip
\textit{Case 1}: Assume that a spanning tree $T$ contains all edges of each star.
We consider an instance with $k=n$ agents initially placed on the leaves of $S_n^0$, that is, $H = \{v_{1}^0, \ldots, v_{n}^0\}$.
In any strategy allowed to use all edges of $G$, only one agent per step can evacuate during the first $s$ steps, since only a single exit is reachable within distance $s$.
During the next $s$ steps, at most two agents per step can evacuate using both $v_0^0$ and $v_0^1$, with the latter exit used by those agents who, during the first $s$ steps, moved along the corresponding paths toward $v_0^1$. Continuing this argument, after $sm$ steps at most $sm(m+1)/2$ agents may have evacuated in total. Clearly, this bound is attainable and hence $OPT=sm$. (See Figure \ref{fig:LBonSpTrees}(b) for an illustration of an optimal strategy).

Now consider any evacuation strategy restricted to the edges of the spanning tree $T$. The key observation is that two distinct agents cannot evacuate through the same exit $v_0^j$ with $j>0$; for this to happen, $T$ would have to contain a cycle. 
As a consequence, evacuating $sm(m+1)/2$ agents requires $A_T \ge k-m+1$ steps for any strategy in this case.
Indeed, during the first $s(m-1)+1$ steps at most $s(m-1)+1$ agents can evacuate through $v_0^0$, while each of $m-1$ exits $v_0^j$ with $j>0$ can be used by at most one agent.
The remaining agents have to evacuate through $v_0^0$ one at a time, which in turn requires at least 
$k-(s(m-1)+1 + (m-1))$ further steps. Thus, 
$A_{T} \ge s(m-1)+1 + (k-(s(m-1)+1+(m-1)) = k-m+1$.
The competitive ratio is hence lower-bounded by
$$
\frac{A_{T}}{\OPT} \ge \frac{sm(m+1)/2-m +1 }{sm} \ge  \frac{m+1}{2} - \frac{1}{s} = \Omega(m).
$$

\textit{Case 2}: Suppose a spanning tree $T$ omits some edges of a star $S_n^j$ for some $j \in \{0,\ldots,m-1\}$; without loss of generality, assume $\{v_0^j,v_1^j\}$ does not belong to $T$. We consider an instance with a single agent initially placed at $v_1^j$ (i.e., $H = \{v_1^j\}$, $k=1$). 

In the optimal strategy, one step suffices to reach the exit at $v_0^j$ and hence $\OPT=1$. 
In contrast, when restricting to $T$, the agent must traverse a path of length at least $s$ to reach any exit. This implies that for the evacuation time $A_T$ along the tree $T$ we have $A_T \ge s$. 
Consequently, the competitive ratio is at least ${A_T}/{\OPT} \ge s=\Theta(\sqrt[5]{N})$.
\end{proof}

% =====================================================================
\section{Core Concepts of the Framework} 
\label{sec:framework}

In this section, we introduce the structure and foundational principles behind our evacuation framework.
We first define the notion of \emph{groups}, which serve as temporary centralized computational units enabling local coordination among agents (cf. Section~\ref{sec:groups}).
Then, we introduce the key concept of a \emph{$\B$-partition}, which groups vertices into coherent regions called \emph{zones}, and culminates in the construction of a \emph{zone graph}---a high-level representation that captures the adjacency relationships between zones (cf. Sections~\ref{sec:framework:partitions_and_zones} and~\ref{sec:framework:zonegraph}). 
By abstracting away fine-grained local details, the zone graph enables agents to coordinate at the zone level and supports the design of phase-based evacuation strategies.
The multi-level framework structure is further enriched in Section~\ref{sec:framework:strategies} by introducing \emph{internal} and \emph{exclusive strategies}, which tailor agents' low-level behavior to the structural characteristics of individual zones.

\medskip
In our framework, we use a nonnegative integer parameter $\B$, interpreted as a time limit that controls the granularity of partitions and thus bounds the duration of evacuation phases. 
Since the agents do not know $\OPT$ a priori, we employ a doubling technique to find an appropriate value for $\B$.
Ultimately, we prove that once $\B\geq\OPT$, the framework guarantees the successful evacuation of all agents within the corresponding phases of the actual epoch.

\medskip
A notable advantage of the framework is that it guarantees a feasible evacuation for any partition into connected subgraphs. 
It follows that while any such partition ensures feasibility, achieving high performance reduces to optimizing the structural properties of the underlying $\B$-partitions (cf. Theorem~\ref{thm:EvacuationTime} and the instantiation of the framework for grids in Section~\ref{sec:grids}).

% ---------------------------------------------------------
\subsection{The Role of Groups} 
\label{sec:groups}

In the distributed model used in this paper, the coordination of agents depends not only on their individual actions but also on the graph's local structure  around them.
Intuitively, a group represents a collection of agents whose neighborhoods overlap in a way that guarantees communication within the group. 
More formally, a set of agents $A$ is called a \emph{group} in a given step if the subgraph induced by $\bigcup_{a\in A}N[a]$ is connected, where $N[a]$ denotes the closed neighborhood of the vertex currently occupied by agent $a$.
The following is a consequence of the computational model.

\begin{proposition} \label{prop:group-communication}
If a set of agents $A$ forms a group at the beginning of a step, then any agents in $A$ can communicate and therefore exchange any number of messages in this~step.
\end{proposition}

Consequently, such groups may act as temporary ``centralized computational units'' coordinated via standard convergecast and broadcast primitives. 
Specifically, all group members can send their local data to a distinguished agent $a$, who then performs a centralized computation taking into account the ``knowledge'' of the entire group. 
The result is a decision regarding the actions in the current step, or a series of actions in several subsequent steps, which is then broadcast by $a$ back to all other agents in the group.

An example of this type of computation is determining the exclusive strategies defined in Section~\ref{sec:framework:strategies}. 
To perform these calculations, we rely on a specific monotonicity property. While this property inherently holds for the \emph{centralized setting} of $\problemEvac$, under some additional assumptions, it can be successfully applied to our distributed approach.

\begin{proposition}\label{prop:monotonicity}
Let $I=(G,\ell,X,H)$ and $I'=(G,\ell,X,H')$ be the instances of the problem $\problemEvac$ with $H'\subseteq H$. If~$I$ is a $\YES$-instance, then $I'$ is also a $\YES$-instance.\qed
\end{proposition}

In fact, when referring to the above context, we usually mean that $\OPT(I')\le \OPT(I)$, which captures a natural intuition that removing agents cannot make the centralized evacuation harder.

% ---------------------------------------------------------
\subsection{$\B$-partitions and Zones}
\label{sec:framework:partitions_and_zones}

Let $\B$ be a nonnegative integer, and let $G[U]$ denote the subgraph induced by a vertex set $U$. 
We say that $(V_1,\ldots,V_l)$ is a \emph{$\B$-partition} of a graph $G=(V,E)$ with the exit set $X$ if the following conditions are satisfied:

\begin{enumerate}[label=(C\arabic*)]
 \item\label{it:partition1} $V_i\cap V_j=\emptyset$ for each $1\leq i<j\leq l$ and $V_1\cup\cdots\cup V_l=V$,
 \item\label{it:partition2} for each $i\in\{1,\ldots,l\}$, there exists a spanning tree of the graph $G[V_i\setminus X]$ 
  with root $v_i\in V_i$ such that for each $v\in V_i$ its distance from $v_i$ along the spanning tree is at most $\B$. 
\end{enumerate}

A subgraph $Z_i=G[V_i]$ induced by $V_i$, $i\in\{1,\ldots,l\}$ is called a \emph{zone}. The vertex $v_i$ is called the \emph{root of the zone $Z_i$}.

\medskip
The role of the spanning tree in Condition~\ref{it:partition2} is to ensure that each agent in $V_i$ can identify the root $v_i$, compute a path to the root along the spanning tree, and traverse this path in at most $\B$ steps.
Using the paths along spanning trees rather than taking arbitrary paths is crucial for bounding the time required to form groups by the agents located in particular zones. 
While an agent may not traverse the entire path due to the presence of other agents, the use of a spanning tree guarantees that all agents in $V_i$ can form a group within $\B$ steps (cf. Lemma~\ref{lem:first-grouping}).

% -------------------------------------------------------
\subsection{Internal and Exclusive Strategies}
\label{sec:framework:strategies}

Suppose we are given a $\B$-partition of a graph $G=(V,E)$, for some nonnegative integer $\B$. 
For the set of agents with their homebases in a zone $Z$ (i.e., initially located in $Z$) denoted by $\agents(Z)$, we introduce two fundamental evacuation strategies that dictate their behavior.

An evacuation strategy is called \emph{internal for} $\agents(Z)$ if during the execution of the strategy each agent in $\agents(Z)$ remains on the vertices of $Z$. 
We say that a zone $Z$ is \emph{self-sufficient} with respect to $\B$ (we omit $\B$ if clear from the context) if there exists a successful internal evacuation strategy that can be computed by the agents in $\agents(Z)$ and completed in $2\B$ steps. Otherwise, we say that $Z$ is \emph{non-self-sufficient}.
Note that, by definition, a self-sufficient zone must contain an exit.

An evacuation strategy is called \emph{exclusive for} $\agents(Z)$ if it is defined for the input instance consisting of the entire graph $G$, only the agents in $\agents(Z)$, and all exits $X$ available exclusively for the evacuation of \emph{only the agents in} $\agents(Z)$.
In other words, it is the instance we obtain from the original one by removing all agents except those initially located in $Z$.

\medskip
We establish the following lemma regarding the calculation of exclusive strategies.
(Note that the evacuation strategy claimed in the lemma assumes that the group of agents has already been formed.)

\begin{lemma} \label{lem:exclusive-strategy}
Suppose that $Z$ is not self-sufficient and the agents in $\agents(Z)$ formed a group within $s$ steps, starting from the homebases.
There exists a successful exclusive evacuation strategy of length at most $s+\OPT$ that can be calculated by the agents in the group.
\end{lemma}

\begin{proof}
There exists the following successful exclusive evacuation strategy that starts in the step after the group is formed. In this strategy, first, all agents in $\agents(Z)$ backtrack the moves they made to form the group, so that they arrive at their initial homebases in $H$; second, they execute a successful strategy $\cS$ of length at most $\OPT$ starting from the homebases.
Since the agents started forming the group from their homebases, this backtracking takes no longer than the duration of group formation, which is at most $s$ steps by assumption.

The successful strategy $\cS$ exists and can be computed via a centralized algorithm.
Indeed, once the agents have formed a group, it can be assumed that one of the agents knows the homebases of all agents in $\agents(Z)$.
Hence, such an agent executes a centralized algorithm that finds an optimal successful evacuation strategy $\cS$ and then sends $\cS$ to all other agents in $\agents(Z)$.
The bound $\OPT$ on the length of $\cS$ follows from Proposition~\ref{prop:monotonicity}. Altogether, our successful exclusive strategy requires at most $s+\OPT$ steps.
\end{proof}

% -----------------------------------------------------------------------
\subsection{Zone Graphs}
\label{sec:framework:zonegraph}

The next component of our framework is a zone graph. 
Given a nonnegative integer $\B$ and a $\B$-partition of a graph $G=(V,E)$, we say that two zones $Z$ and $Z'$ are \emph{close} if $Z$ contains a vertex $z\in V(G)$ and $Z'$ contains a vertex $z'\in V(G)$ such that
\begin{equation}\label{eq:dist}
   \dist{z}{z'}\leq 2\B,
\end{equation}
where $\dist{u}{v}$ for two vertices $u$ and $v$ is the length of a~shortest $uv$-path in $G$ that \emph{does not contain an exit}.

For a $\B$-partition $(V_1,\ldots,V_l)$ of $G$, let $Z_i=G[V_i]$, $i\in\{1,\ldots,l\}$, denote the \emph{$i$-th zone}.
The \emph{zone graph} of a $\B$-partition $(V_1,\ldots,V_l)$ is a graph $\ZG=(\ZGV,\ZGE)$ with the vertex set $\ZGV=\{Z_1,\ldots,Z_l\}$ and the edge set $\ZGE$ such that $\{Z,Z'\}\in\ZGE$ if and only if the zones $Z$ and $Z'$ are close and none of them is self-sufficient.

By the above definition, the last condition implies that each self-sufficient zone is an isolated vertex in the zone graph $\ZG$.
Accordingly, we say that the zones $Z_1,\ldots,Z_p$ are \emph{independent} if $\{Z_i,Z_j\}\notin\ZGE$ for any $1\leq i<j\leq p$. 

\medskip
As shown in the subsequent analysis, the chromatic number of the zone graph multiplicatively contributes to the evacuation time. Modeling self-sufficient zones as an independent set of isolated vertices in the corresponding zone graph ensures that their contribution to the evacuation time is minimized.

% -----------------------------------------------------------------
\subsection{Framework Operation Outline} \label{sec:framework:outline}

Before the formal analysis, we give a high-level overview of how the framework operates. 
The entire strategy is divided into epochs, composed of several phases as reflected in the pseudocode of Algorithms \ref{alg:Evacuation}, \ref{alg:OneAttempt}, and \ref{alg:SinglePhase}.
The sequence of steps performed by the agents resulting from a single call to Algorithm \ref{alg:SinglePhase} constitutes one \emph{phase}.
Correspondingly, a single call to Algorithm \ref{alg:OneAttempt} for a fixed value of $\B$ constitutes one \emph{epoch}.
In each epoch, the goal is to compute and execute a successful evacuation strategy for all agents.

\begin{algorithm}[t!]
\SetAlgoRefName{{\upshape\scshape Evacuation}}
   \caption{\\Input: a graph $G$ with an exit set $X$}
   \label{alg:Evacuation}
   $\B=2$
   
   \While{there are agents present in the graph $G$}{
      Compute a $\B$-partition $(V_1,\ldots,V_l)$ \label{ln:evac-partition}
      
      Call Algorithm~\ref{alg:OneAttempt}($\B,(V_1,\ldots,V_l)$) \label{ln:evac-phase}
      
      $\B=2\B$
   }
\end{algorithm}

To achieve this, Algorithm \ref{alg:Evacuation} repeatedly launches attempts with successively doubled values of the parameter $\B$ (a bound $\B=2^j$ for $\OPT$ in the $j$-th epoch), which controls the granularity of the calculated partitions. For each attempt, the agents compute the corresponding $\B$-partition of the graph and execute Algorithm \ref{alg:OneAttempt} (with this partition).
As we show later, once $\B$ becomes sufficiently large (specifically, $\B\ge\OPT$), a successful evacuation strategy is guaranteed to be found and completed in this epoch. Conversely, if $\B<\OPT$, then the evacuation attempt in such an epoch may fail, i.e., not all agents evacuate. The agents can easily detect such an unsuccessful attempt, and hence every agent that did not evacuate within the assumed time limit backtracks and starts over from its homebase in the next epoch. The algorithm continues by doubling~$\B$ before the next attempt.

\begin{algorithm}[b!]
\SetAlgoRefName{{\upshape\scshape OneAttempt}}
   \caption{\\Input: an integer $\B$ and a $\B$-partition $(V_1,\ldots,V_l)$ of the graph $G$}
   \label{alg:OneAttempt}
   Compute the zone graph $\ZG$ of the $\B$-partition $(V_1,\ldots,V_l)$
   
   Find a $d$-coloring $c$ of $\ZG$
   
   \ForEach{$i\in\{1,\ldots,d\}$}{
      Call \ref{alg:SinglePhase}($\B, c^{-1}(i)$) \label{ln:OneAttempt:SinglePhase}

      \mbox{Backtrack all not evacuated agents in each $Z$ with $c(Z)=i$ to their homebases}\label{ln:OneAttempt:backtrack}
   }
\end{algorithm}

We emphasize that backtracking to the starting position after each phase is a key design choice that significantly simplifies the analysis, while adding only a constant factor to the final competitive ratio.
The same justification applies to the doubling technique, which also contributes only a constant factor to the evacuation time. 
As our primary objective is to establish a general framework allowing a constant competitive ratio, we prioritize structural simplicity over the optimization of exact constant factors.

\medskip
A single attempt, captured by Algorithm \ref{alg:OneAttempt}, takes the computed $\B$-partition and constructs the associated zone graph $\ZG=(\ZGV,\ZGE)$.
The computation of the zone graph $\ZG$ can be done without knowing the locations of other agents, and therefore, it is performed independently by each agent in the first step of each epoch.
(Recall that the definition of the zone graph depends on $\B$ and thus the zone graph is different in each epoch).
Next, applying the same deterministic algorithm, all agents find an identical $d$-\emph{coloring} of $\ZG$ defined as a function $c\colon\ZGV\rightarrow\{1,\ldots,d\}$ such that $c(Z)\neq c(Z')$ for any adjacent zones $Z$ and $Z'$.
Consequently, since each agent knows $\B$ at the beginning of an epoch, it can locally determine the number of phases $d$.

Algorithm~\ref{alg:OneAttempt} coordinates the agents at the level of zones, leaving the vertex-level coordination to Algorithm \ref{alg:SinglePhase} that actually handles agents' movements realizing internal or exclusive strategies of particular zones.

\begin{algorithm}[tb!]
\SetAlgoRefName{{\upshape\scshape SinglePhase}}
   \caption{\\Input: an integer $\B$ and a set of zones $\cZ$}
   \label{alg:SinglePhase}
      Let $t$ be the step to be executed \label{ln:phase:start}
      
      \ForEach{$Z\in\cZ$}{
         \If{$Z$ is self-sufficient}{
            Compute an internal strategy $\cS$ for $\agents(Z)$ \label{ln:phase:internal-strategy}
            
            Execute $\cS$ in steps $t,\ldots,t+2\B-1$
            \label{ln:phase:execute-internal-strategy}
         }\Else{
            Group the agents $\agents(Z)$ in steps $t,\ldots,t+\B-1$\label{ln:phase:grouping}
            
            Compute an exclusive strategy $\cS$ for $\agents(Z)$ starting at the agents'  current locations \label{ln:phase:exclusive-strategy}
            
            Execute $\cS$ in steps $t+\B,\ldots,t+3\B-1$ \label{ln:phase:execute-exclusive-strategies}
      }
   }
\end{algorithm}

\medskip
In each phase $i\in\{1,\ldots,d\}$, Algorithm \ref{alg:SinglePhase} attempts to evacuate agents in zones colored $i$, leaving other agents idle. Such zones are handled in two distinct ways: 

\begin{enumerate}
\item The agents in self-sufficient zones compute successful internal evacuation strategies and execute them. 
Our framework abstracts away the exact method of finding these internal strategies (it may greatly depend on the structure of the input graph and on the locations of the exits), but assumes that for each self-sufficient zone such a strategy exists (see the definition of self-sufficient zone) and is provided in an application of the framework (which we demonstrate for grids in Section~\ref{sec:grids}).

\item The agents in other zones compute their strategies in a distributed manner provided explicitly by the framework (see, e.g., Lemma \ref{lem:exclusive-strategy}).
We emphasize that handling these zones constitutes the core algorithmic difficulty in our computation.
Furthermore, the computed strategies evacuate the agents not only using the exits in the zone, but also outside the zone.
\end{enumerate}

Ultimately, in every execution, after several epochs that fail, the framework reaches an epoch in which $\B\geq\OPT$. This ensures that a successful evacuation strategy is computed in the current epoch.

% ================================================================
\section{Analysis of the Framework} \label{sec:framework:analysis}

The main result of this section is Theorem~\ref{thm:EvacuationTime}.
Its proof is deferred to Section~\ref{sec:prf}, following the formal analysis of the framework's components and their properties.

\begin{theorem} \label{thm:EvacuationTime}
If $d_j$ is the number of colors used for the zone graph  computed in the $j$-th epoch, then all agents evacuate in $s$ steps, where
\begin{equation}
    s=4\sum_{j=1}^{p}d_j2^j
\end{equation}
and $p=\lceil\log_2\OPT\rceil$.
\end{theorem}

% --------------------------------------------------------------
\subsection{A Single Phase} 
\label{sec:framework:phase}

We first analyze Algorithm \ref{alg:SinglePhase}, which handles the evacuation of agents within a given set of zones $\cZ$.
As we will see, $\cZ$ constitutes a specific subset of zones derived from the underlying $\B$-partition, allowing evacuations in particular zones to be scheduled concurrently. The algorithm explicitly manages potential inter-zone interactions, with the objective of evacuating all agents in each zone $Z\in\cZ$ in at most $3\B$ steps.

To achieve this, Algorithm \ref{alg:SinglePhase} distinguishes between self-sufficient and non-self-sufficient zones.
If $Z$ is self-sufficient, then agents simply execute an internal evacuation strategy taking $2\B$ steps $t,\ldots,t+2\B-1$. 
Conversely, for a non-self-sufficient zone $Z$, the first $\B$ steps are used for grouping the agents (see line~\ref{ln:phase:grouping}), a mechanism whose correctness we establish in Lemma~\ref{lem:first-grouping}.
Subsequently, the agents in $\agents(Z)$ compute and execute an exclusive strategy $\cS$ (see line~\ref{ln:phase:exclusive-strategy}). The strategy $\cS$ operates under the premise that \emph{agents originate from their current locations} (i.e., their locations at step $t+\B-1$) rather than their initial homebases (cf. Lemma~\ref{lem:exclusive-strategy})

To ensure that the concurrent execution of these strategies incurs no delays, the framework handles two types of potential interactions.
First, the strategy $\cS$ may use an exit $x\notin Z$ to evacuate some agents in $\agents(Z)$. However, we will argue that, within an epoch, agents from different zones $Z$ in $\cZ$ never attempt to use such an exit $x$ at the same time.
Second, moving agents executing $\cS$ might encounter idle agents from self-sufficient or neighboring zones in $\ZG$. In this case, the duration of $\cS$ cannot increase because moving agents can skip the idle ones.
To see this, consider an agent $a$ located at a vertex $v_0$ that must traverse a path $v_0,\ldots,v_p$ whose internal vertices $v_1,\ldots,v_{p-1}$ are already occupied. If these agents remain idle for $p$ steps, then in step $p$ each agent at $v_i$ for $i\in\{0,\ldots,p-1\}$ moves to $v_{i+1}$, with its memory contents transferred via message exchanges, so that the agent at $v_p$ receives the data of $a$, including its identifier. This simulates $a$ moving from $v_0$ to $v_p$ in $p$ steps while the other agents remain idle, a process we refer to as a \emph{skip} of the agents at $v_1,\ldots,v_{p-1}$ by the agent $a$.

\begin{lemma} \label{lem:first-grouping}
Consider a zone $Z$ of a $\B$-partition of $G$.
If $Z$ is not self-sufficient, then there exists an algorithm that groups the agents in $\agents(Z)$ in at most $\B$ steps.
\end{lemma}

\begin{proof}
Condition~\ref{it:partition2} in the definition of $\B$-partition distinguishes a vertex $v\in Z$ such that any other vertex $u$ of $Z$ is connected to $v$ by a path $P_{uv}$ of length at most $\B$ (the existence of such a vertex, the root of $Z$, is guaranteed by the definition of a zone).
Consider a grouping strategy in which each agent $a$ currently at $u$ attempts to traverse $P_{uv}$ from $u$ to $v$ (recall that we consider the paths along a fixed spanning tree).
The moves of the agent are defined such that if $a$ is able to get closer to $v$ (moving along $P_{uv}$), then it makes the move; otherwise it remains idle.
In this context, observe that Condition~\ref{it:partition2} implies that whenever the paths of agents starting at distinct vertices $u$ and $u'$ have a common vertex $v'$, then they have to share a subpath $P_{v'v}$ connecting $v'$ with $v$.
Consequently, whenever $a$ needs to stay idle in a given step, then either all vertices of the path $P_{uv}$ are occupied by agents
or there is some agent occupying a vertex of $P_{uv}$ that gives a way to an agent $a'$ on the other path that shares with $P_{uv}$ a subpath with the end in $v$. Without loss of generality, we may assume that this is the path $P_{u'v}$. Therefore, the number of \emph{unoccupied vertices} on $P_{v'v}$ cannot increase (recall that the endvertex $v$ is the root) and since the length of $P_{v'v}$ is at most $\B$, an inductive argument shows that $\B$ steps are sufficient to form a group consisting of all agents in $\agents(Z)$.
\end{proof}

Under the assumption that $\B\ge\OPT$, in Lemma~\ref{lem:SinglePhase:S} we establish that $3\B$ steps of Algorithm~\ref{alg:SinglePhase} suffice to evacuate agents from any zone that is not self-sufficient.
Recall that, in general, agents in such zones may need to evacuate via exits located outside the zone.
To accurately capture the coordination required for inter-zone evacuations, we define a zone $Z$ as \emph{idle} in a given phase if $Z\notin\cZ$, where $\cZ$ is the input provided to Algorithm~\ref{alg:SinglePhase}.

\begin{lemma} \label{lem:SinglePhase:S}
Let $\B$ and $\cZ$ be an input to Algorithm \ref{alg:SinglePhase}.
If $\B\geq\OPT$, the zones in $\cZ$ are independent, all zones not in $\cZ$ are idle in this phase, and each self-sufficient zone $Z$ in $\cZ$ admits an internal successful evacuation strategy of length at most $2\B$, then the strategy computed and executed by the agents in $\agents(Z)$ is successful for each $Z\in\cZ$.
\end{lemma}

\begin{proof}
Consider any two zones $Z, Z' \in\cZ$. 
The proof naturally falls into three cases depending on self-sufficiency of $Z$ and $Z'$.

First, consider the case where both zones are self-sufficient. 
The strategies defined for them in line~\ref{ln:phase:internal-strategy} are internal and hence both may be correctly executed in parallel during steps $t,\ldots,t+2\B-1$.
(The fact that such a strategy exists for each self-sufficient zone and can be computed independently by the agents initially located in the zone follows by assumption of the lemma.)

Second, if neither $Z$ nor $Z'$ is self-sufficient, then the independence of the zones in $\cZ$ implies that $Z$ and $Z'$ are not close (recall that by \eqref{eq:dist} this yields $\dist{z}{z'}>2\B$ for any $z\in Z$ and $z'\in Z'$).
Since $\B\geq\OPT$, the definition of the zone graph guarantees that in every optimal successful strategy, the paths traversed by any agents $a\in\agents(Z)$ and $a'\in\agents(Z')$, share no vertex; in particular, they end at different exits.
Take any optimal successful evacuation strategy $\tilde{\cS}$ and consider the exclusive strategies $\cS$ and $\cS'$ obtained by restricting $\tilde{\cS}$ to the agents in $\agents(Z)$ and $\agents(Z')$, respectively.
Because in the strategy $\tilde{\cS}$ no agent in $\agents(Z)$ communicates with an agent in $\agents(Z')$, the strategies $\cS$ and $\cS'$ inherently preserve this property.
Accordingly, this extends to the strategies computed in line~\ref{ln:phase:exclusive-strategy} during the corresponding iterations of the main loop of Algorithm \ref{alg:SinglePhase}.
Applying Lemma~\ref{lem:exclusive-strategy} (with $s\leq\B$), we conclude that the strategy computed in line~\ref{ln:phase:exclusive-strategy} of Algorithm \ref{alg:SinglePhase} is successful and requires at most $2\B$ steps. 
Thus, the time limit established in line~\ref{ln:phase:execute-exclusive-strategies} is sufficient for the strategy to be correctly executed.

Finally, we consider the case in which $Z$ is self-sufficient and $Z'$ is not.
By Lemma \ref{lem:first-grouping}, the initial $\B$ steps $t,\ldots,t+\B-1$ are sufficient for the agents in $Z'$ to group and calculate an exclusive strategy $\cS$ (see line~\ref{ln:phase:grouping}). 
Next, in steps $t+\B,\ldots,t+2\B-1$, the agents in $Z'$ return to their homebases.
Since no agent leaves $Z'$ during the entire process in steps $t,\ldots,t+2\B-1$, grouping and backtracking in $Z'$ can be executed concurrently with the internal evacuation strategy for $Z$. 
Moreover, by the assumption of the lemma regarding internal strategies, all agents in $Z$ successfully evacuate by the end of step $t+2\B-1$ (due to $\B\geq\OPT$). 
Therefore, the exclusive strategy $\cS$ executed in line~\ref{ln:phase:execute-exclusive-strategies} cannot be influenced by any agent in $\agents(Z)$.
\end{proof}

Since in Algorithm~\ref{alg:SinglePhase} the strategy executed in line~\ref{ln:phase:execute-internal-strategy} requires $2\B$ steps and the strategies in lines~\ref{ln:phase:grouping} and~\ref{ln:phase:execute-exclusive-strategies} collectively require $3\B$ steps, Lemma~\ref{lem:SinglePhase:S} yields the following bound.

\begin{lemma} \label{lem:singlePhaseTime}
An execution of Algorithm \ref{alg:SinglePhase} with input $\B$ results in $3\B$ steps used by an evacuation strategy.
\qed
\end{lemma}

% ----------------------------------------------------------------------
\subsection{Combining all Phases within an Epoch} \label{sec:framework:phases}

In this section, we analyze the unified strategy for a single epoch, which results from combining individual phases---each corresponding to a specific color in the zone graph---into a single execution of Algorithm~\ref{alg:OneAttempt}.

\begin{lemma} \label{lem:colored-zones}
Suppose that the zone graph $\ZG$ admits a $d$-coloring.
Then, Algorithm \ref{alg:OneAttempt} computes an evacuation strategy that takes $4d\B$ steps.
If $\B\geq\OPT$, then the strategy is successful.
If $\B<\OPT$, then the strategy either evacuates an agent or leaves it at its homebase.
\end{lemma}

\begin{proof}
By Lemma~\ref{lem:singlePhaseTime}, Algorithm \ref{alg:SinglePhase} takes $3\B$ steps each time it is called in line \ref{ln:OneAttempt:SinglePhase} of \ref{alg:OneAttempt}.
Next, the backtracking in line~\ref{ln:OneAttempt:backtrack} takes at most $\B$ steps.
This holds because the time required for backtracking cannot exceed the time agents spend moving from their homebases towards the exits during the exclusive evacuation strategy $\cS$ (determined in line~8 of Algorithm~\ref{alg:SinglePhase}).
Specifically, $\cS$ first backtracks agents from their current locations (resulting from the grouping in line~\ref{ln:phase:grouping}) back to their homebases, and then executes an evacuation phase of duration $\B$.
Altogether, since a single iteration takes $4\B$ steps and there are $d$ calls to Algorithm \ref{alg:SinglePhase}, the time bound of $4d\B$ follows.

Moreover, the backtracking guarantees that if $\B<\OPT$, then each agent still present in the graph\footnote{Note that even if $\B<\OPT$, the evacuation may succeed in such an epoch. This may happen for two reasons. First, the duration of the epoch is $3\B$ and thus may exceed $\OPT$. Second, some agents potentially evacuated in earlier epochs and thus the optimal evacuation time may be smaller than $\OPT$ for the agents currently present in the graph.} at the end of the execution of Algorithm \ref{alg:SinglePhase} in line~\ref{ln:OneAttempt:SinglePhase} returns to its homebase.
Furthermore, because $c$ is a coloring of $\ZG$, all zones in a single color class $c^{-1}(i)$ are independent. Therefore, if $\B\geq\OPT$, then Lemma~\ref{lem:SinglePhase:S} implies that the call to Algorithm \ref{alg:OneAttempt} yields a successful evacuation strategy. 
\end{proof}

% -----------------------------------------------------------------
\subsection{Efficiency of the Framework --- Proof of Theorem \ref{thm:EvacuationTime}}
\label{sec:prf}

In this section, we finally establish the main result regarding the bound on the evacuation time achieved by Algorithm \ref{alg:Evacuation}.

\begin{proof}[Proof of Theorem~\ref{thm:EvacuationTime}]
The execution of Algorithm~\ref{alg:Evacuation} consists of a sequence of epochs $j\in\{1,\dots,p\}$, where $p=\lceil\log_{2}\OPT\rceil$ is the smallest integer such that $2^p\geq\OPT$. 
By Lemma~\ref{lem:colored-zones}, the $j$-th epoch is executed in $4d_j 2^j$ steps, where $d_j$ denotes the chromatic number of the zone graph of a $\B$-partition for $\B=2^j$.
Naturally, for each $j<p$, Lemma~\ref{lem:colored-zones} guarantees that any agent who fails to evacuate in epoch $j$ returns to its initial homebase, ensuring a consistent state for the subsequent epoch.
In particular, by Proposition~\ref{prop:monotonicity}, this ensures that the optimal evacuation time for the instance at the beginning of each epoch is upper bounded by $\OPT$.
Finally, for $j=p$, the lemma guarantees that the evacuation is successful for all remaining agents.
Summing the durations of all epochs yields the final formula for the evacuation time.
\end{proof}

% =====================================================================
\section{Asymptotically Optimal Evacuation in Grids} 
\label{sec:grids}

We now instantiate the framework for two-dimensional grids. An $n \times m$ \emph{grid} $G=(V,E)$ is defined as the Cartesian product of two paths, where the vertex set $V$ consists of integer coordinate pairs $(i,j)$ with $1 \le i \le n$ and $1 \le j \le m$, and the edge set $E$ consists of all sets $\{(i,j), (i',j')\}$ of vertices at unit distance, i.e., $|i-i'| + |j-j'| = 1$. 
In this coordinate system, $(i,j)$ denotes the vertex in the $i$-th column and $j$-th row, while $(1,1)$ and $(n,m)$ represent the lower-left and upper-right corners, respectively.

\medskip
As the central result of this section, the following theorem establishes that our framework guarantees asymptotically optimal evacuation strategies for grids.
\begin{theorem} \label{thm:grids}
There exists an $\bigo(1)$-competitive distributed successful evacuation strategy for each finite $2$-dimensional grid.
\end{theorem}

In order to prove Theorem~\ref{thm:grids}, we must show how to construct a valid $\B$-partition and the corresponding zone graph, and describe appropriate strategies. 
To this end, our approach relies on covering the grid with specific subgraphs called $\B$-areas.
Namely, for any $\B\geq 2$, a \emph{$\B$-area} is a subgraph of a grid $G$ induced by the vertices $(i,j)$ such that $a\B/2 < i \leq(a+1)\B/2$ and $b\B/2 < j \leq(b+1)\B/2$ for some $a,b\geq 0$.
Intuitively, each $\B$-area is a $\B/2\times\B/2$ sized square subgrid with the lower-left vertex $(a\B/2+1,b\B/2+1)$. 
If $m$ or $n$ is not a multiple of $\B/2$, then some peripheral areas of the grid cannot be covered. However, in such a case one can either artificially increase the dimensions of the grid so that they become multiples of $\B/2$ (by adding vertices that are neither homebases nor exits), or alternatively, observe that the analysis can be easily extended to rectangular $\B$-areas.

The core of our approach relies on dividing each individual $\B$-area $H$ into appropriate zones, ultimately allowing us to obtain a $\B$-partition of the whole grid.
Ideally, due to its bounded diameter, an entire $\B$-area $H$ could act as a single zone. However, the locations of exits might hinder forming a group of agents in a sufficiently short time. 
While treating exits as obstacles for group formation might seem counterintuitive, this approach allows us to apply our framework directly as a black-box. More importantly, whenever the exits prevent $H$ from becoming a single zone, this very obstruction allows us to create self-sufficient zones. In such zones, agents can execute simple, uninterrupted evacuation paths directly to the exits.

\medskip
Technically, to formalize this partitioning, we consider a local coordinate system for each $\B$-area $H$, denoting its vertices by $(i,j)$ for $1\leq i,j\leq\B/2$, effectively renaming the indices so that its lower-left vertex is $(1,1)$. 
The partitioning of $H$ depends heavily on the existence of specific paths within this local system.
We say that a path $P\subseteq H$ is \emph{monotone} if it contains no exits and $(1,j)\in V(P)$ and $(\B/2,j')\in V(P)$ for some $j$ and $j'$, and for any two consecutive vertices $(i,z)$, $(i',z')$, we have $i' \geq i$ and $z' \geq z$.

The presence or absence of such a path dictates how we structurally organize the area around the exits to obtain the desired zones, as established by the following two lemmas (an example of such a partition is illustrated in Figure~\ref{fig:monotPath}).

\begin{lemma} \label{lem:no-monotone-path}
If a $\B$-area $H$ does not contain a monotone path, then each row of $H$ is a self-sufficient zone.
\end{lemma}

\begin{proof}
Observe that if $H$ does not contain a monotone path, then this in particular means that for each $j\in\{1,\ldots,\B/2\}$ one of the vertices in the $j$-th row $R_j=\{(1,j),\ldots,(\B/2,j)\}$ is an exit.
Indeed, if none of the vertices in $R_j$ is an exit, then the row fulfills the definition of a monotone path.

We prove that the subgraph, denoted by $Z$, induced by the vertices in $R_j$ is a self-sufficient zone.
This is done by arguing that an appropriate internal strategy for $Z$ exists.
The fact that $Z$ contains an exit $x$ makes the following evacuation strategy successful: each agent in $\agents(Z)$ simply moves towards $x$ and evacuates through $x$ (or potentially through another vertex if more exits are present in $R_j$).
Since $|\agents(Z)|\leq|V(Z)|=|R_j|\leq\B/2$, the duration of this strategy is at most $\B/2$, which completes the proof.
\end{proof}

We point out that checking the existence of a monotone path and computing the strategy implicit in the proof of Lemma~\ref{lem:no-monotone-path} can be performed by each agent based solely on its input to the problem. Since no communication is required, these internal strategies can be executed in at most $\B/2$ steps. 
Furthermore, since the grid topology and exit locations are known, all agents deterministically compute the same monotone path.

\begin{figure}[tbp]
\centering
\includegraphics[width=0.64\textwidth]{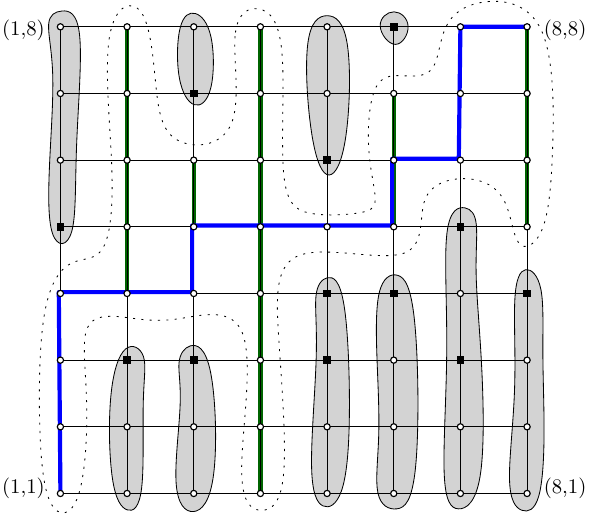}
\caption{An illustration for Lemma~\ref{lem:monotone-path}. A partition of a $\B$-area: a monotone path (blue), the spanning tree $T_H$ (blue+green), the exits (black squares), self-sufficient zones (gray areas), unique non-self-sufficient zone (dotted border area).}
\label{fig:monotPath}
\end{figure}

\begin{lemma} \label{lem:monotone-path}
If a $\B$-area $H$ contains a monotone path, then there exists a $\B$-partition of $H$ with exactly one non-self-sufficient zone.
\end{lemma}

\begin{proof}
We construct a spanning tree $T_H$ of $H$ as follows. First, let $P$ be a monotone path in $H$ and let $P$ belong to $T_H$.
Then, for each column $i\in\{1,\ldots,\B/2\}$, add to $T_H$ the path induced by the vertices $(i,j)$ with $j\in\{1,\ldots,\B/2\}$ (intuitively, we use all ``vertical'' paths, which together with $P$ form a spanning tree $T_H$).

This spanning tree $T_H$ provides the zones contained in $H$ in the following way.
The unique zone that is not self-sufficient, denoted by $Z$, is defined as the maximal subtree of $T_H$ that contains $P$ and has no exits.
Note that such a zone $Z$ fulfills Condition~\ref{it:partition2}, where the required spanning tree is a subtree of $T_H$ with the root in $P$. 
A successful strategy for zone $Z$ is provided by the framework (cf. Lemma~\ref{lem:exclusive-strategy}).

The rest of the $\B$-area is partitioned into self-sufficient zones as follows.
Let $C_i$ denote the $i$-th column of $H$.
Any other maximal subset of $C_{i}\setminus V(Z)$ that induces a connected subgraph of $H$ is defined as a self-sufficient zone (note that due to maximality each exit belongs to some self-sufficient zone).

The construction yields two facts: in accordance with the definition of self-sufficiency, each of our self-sufficient zones has an evacuation strategy of length at most $2\B$, and all these strategies use pairwise distinct exits for evacuation. 
Consequently, they can be executed simultaneously (as in Algorithm~\ref{alg:OneAttempt}). 
\end{proof}

Having established how individual $\B$-areas are partitioned, we now generalize this to the entire grid to bound the chromatic number of the resulting zone graph.

\begin{lemma} \label{lem:grid-coloring}
Let $G$ be a grid.
There exists a $\B$-partition of $G$ such that the corresponding zone graph has a $\Theta(1)$-coloring. 
\end{lemma}

\begin{proof}
For each $\B$-area $H$ of $G$, we apply either Lemma~\ref{lem:no-monotone-path} or~\ref{lem:monotone-path} to obtain a local $\B$-partition of $H$ with at most one non-self-sufficient zone, denoted by $Z(H)$. 
Then, the $\B$-partition of the entire grid $G$ is simply defined as the union of all zones in these local partitions. 

In order to bound the chromatic number of the corresponding zone graph, recall that the vertices representing self-sufficient zones are isolated and hence do not affect the number of colors. Consequently, we focus on non-self-sufficient zones. 
Let $Z(H)$ be such a zone in an arbitrarily selected $\B$-area $H$.
Consider the zones $Z'$ that are close to $Z(H)$. Naturally, for any non-self-sufficient zone in a grid, the number of its neighbors in the zone graph is bounded by a constant that follows from the underlying grid topology. Consequently, the maximum degree $\Delta$ of the zone graph is constant. The lemma follows, since any graph admits a coloring with at most $\Delta+1$ colors.
\end{proof}

The number of colors in the above lemma is bounded by the same constant regardless of the value of $\B$ we select.
Therefore, the number of colors does not change across subsequent epochs in the case of grids.

\begin{proof}[Proof of Theorem~\ref{thm:grids}]
First, we run Algorithm \ref{alg:Evacuation} using $\B$-partitions guaranteed by Lemma~\ref{lem:grid-coloring}.
Hence, for each $\B\in\{2^j\colon j\geq 1\}$, the number of colors used for the zone graph computed in the $j$-th epoch is bounded by a constant $d_j=\Theta(1)$.
By Theorem~\ref{thm:EvacuationTime}, this gives the following bound on the evacuation time
\begin{equation}
    \Theta(1)\cdot\sum_{j=1}^{p}2^j=\Theta(1)\cdot 2^p=\Theta(\OPT),
\end{equation}    
because $p=\lceil\log_2\OPT\rceil$.
\end{proof}

% =======================================================================
\section{Conclusions and Open Problems}
\label{sec:conclusion}

In this paper, we introduced a general algorithmic framework for distributed discrete evacuation on graphs. Relying on the concepts of $\B$-partitions, zone graphs, and phased execution, we established a systematic approach that yields provably correct and competitive strategies. As our main application, we provided an $\bigo(1)$-competitive distributed evacuation strategy for grid graphs. This approach achieves an asymptotically optimal $\Theta(\OPT)$ evacuation time without relying on centralized coordination or global knowledge of the agents' initial configuration.

Our findings open several intriguing directions for future research.

First, while we instantiated our framework for regular grids, the same approach applies to triangular and hexagonal grids.
We mention these topologies as they are frequently used to model programmable matter and robotic swarms.
A natural next step would be an extension of this approach to \emph{partial grids}, i.e., arbitrary subgraphs of grids. Furthermore, adapting our framework to dynamic network topologies, where communication paths may change to reflect blocked evacuation passages, is a highly relevant direction.

From a structural perspective, we proved that strategies restricting to a single spanning tree cannot be constant-competitive in general graphs. However, it remains an open question whether such a spanning tree approach could succeed for bounded-degree graphs. This is particularly interesting given that similar approaches have been successfully applied in related graph-searching problems (see, e.g., Hollinger et al.~\cite{HolKeh10} and Katsilieris et al.~\cite{Katsilieris10}).
Moreover, since evacuation strategies restricted to a spanning tree pre-computed in the first step (see  Proposition~\ref{prop:spanning-tree}) inherently fail to achieve a constant competitive ratio, an interesting open question is whether it also holds if we relax this assumption. For instance, in an alternative approach, agents might postpone the decision on the spanning tree selection, dynamically ensuring only that all traversed edges eventually form an acyclic subgraph. Analyzing such adaptive tree-based strategies remains challenging since, for certain input instances, initial routing decisions must be made before agents can communicate, even indirectly. 

Regarding the underlying model, we assumed the homogeneity of the agents. An intriguing open question is what can be achieved in the distributed setting by heterogeneous agents. Furthermore, it is worth studying the computability of evacuation strategies by agents with weaker capabilities, such as limited computational power, bounded memory size, asynchronous movements, or restricted communication bandwidth.

Finally, the structural abstraction of $\B$-partitions, zone graphs, and phased execution does not seem specific to evacuation alone. Whether the framework developed here can be applied to other distributed agent coordination problems remains a compelling direction for future work.

% ===============================================================

\bibliographystyle{plainurl}
\bibliography{evac-base}

\end{document}